\documentclass[10pt]{llncs}
\def\nottoobig#1{{\hbox{$\left#1\vcenter to1.111\ht\strutbox{}\right.\n@space$}}}
\if01

\newtheorem{theorem}{Theorem}[section]

\newtheorem{definition}[theorem]{Definition}

\fi
\newtheorem{myclaim}[theorem]{Claim}

\usepackage{epsfig}
\usepackage{amsmath}
\usepackage{amsfonts}

\newcommand{\e}{{\rm E}}

\newcommand{\nat}{{\mathbb N}}

\newcommand{\p}{{\rm P}}

\newcommand{\np}{{\rm NP}}

\newcommand{\poly}{{\rm poly}}

\def\nottoobig#1{{\hbox{$\left#1\vcenter
to1.111\ht\strutbox{}\right.\n@space$}}}

\newcommand{\prob}{{\rm Prob}}
\newcommand{\ie}{$\mbox{i.e.}$}

\newlength{\filength}
\newsavebox{\gcbox}
\sbox{\gcbox}{\framebox[\filength]{\rule{0ex}{2ex}}}

\newcommand{\qedblob}{\mbox{\rule[-1.5pt]{5pt}{10.5pt}}}
\def\literalqed{{\ \nolinebreak\hfill\mbox{\qedblob\quad}}}

\def\qed{\literalqed}

\newcommand{\singlespacing}{\let\CS=
\@currsize\renewcommand{\baselinestretch}{1}\tiny\CS}
\newcommand{\singlespacingplus}{\let\CS=
\@currsize\renewcommand{\baselinestretch}{1.25}\tiny\CS}
\newcommand{\doublespacing}{\let\CS=
\@currsize\renewcommand{\baselinestretch}{1.75}\tiny\CS}
\newcommand{\draftspacing}{\let\CS=
\@currsize\renewcommand{\baselinestretch}{2.0}\tiny\CS}

\def\zo{\{0,1\}}

\def\mapping{\rightarrow}

\usepackage{amsfonts}

\if01
\fi

\if01
\fi

\if
\fi
\makeatletter
\def\@listI{\leftmargin\leftmargini \parsep 4.5pt plus 1pt minus 1pt\topsep6pt plus 2pt minus 2pt \itemsep  2pt plus 2pt minus 1pt}

\let\@listi\@listI
\@listi
\makeatother

\newcommand{\bn}{B^{=n}}
\newcommand{\cd}{\rm{CD}}
\newcommand{\cnd}{\rm{CND}}

\author{
{Marius Zimand}
\thanks{ { The author is supported in part
by NSF grant CCF 1016158. URL: \tt  http://triton.towson.edu/\~{ }mzimand}.}}
\institute{
{Department of Computer and Information Sciences, Towson University,
Baltimore, MD, USA}
}

\date{ }

\title{On the optimal compression of sets in PSPACE}

\begin{document}

\maketitle

\begin{abstract} We show that if ${\rm DTIME}[2^{O(n)}]$ is not included in ${\rm DSPACE}[2^{o(n)}]$, then, for every set $B$ in PSPACE, all strings $x$ in $B$ of length $n$ can be represented by a string $compressed(x)$ of length at most $\log (|\bn|) + O(\log n)$, such that a polynomial-time algorithm, given $compressed(x)$, can distinguish $x$ from all the other strings in $\bn$. 
Modulo the $O(\log n)$ additive trem, this achieves the information-theoretical optimum for string compression. 
\end{abstract}

{\bf Keywords:} compression, time-bounded Kolmogorov complexity, pseudo-random generator.
\smallskip

\section{Introduction}
In many practical and theoretical applications in computer science, it is important to represent information in a compressed way. If an application handles strings $x$ from a finite set $B$, it is desirable to represent every $x$ by another shorter string $compressed(x)$ such that $compressed(x)$ describes unambigously the initial $x$. Regarding the compression rate, ideally, one would like to achieve the information-theoretical bound $|compressed(x)| \leq \log (|B|)$, for all $x \in B$. If a set $B$ is computably enumerable, a fundamental result in Kolmogorov complexity states that for all $x \in \bn$, $C(x) \leq \log(|\bn|) + O(\log n)$, where $C(x)$ is the Kolmogorov complexity of $x$, \ie, the shortest effective description of $x$ ($\bn$ is the set of strings of length $n$ in $B$). The result holds because $x$ can be described by its rank in the enumeration of $\bn$. However enumeration is typically a slow operation and, in many applications, it is desirable that the unambiguous description is not merely effective, but also efficient. This leads to the idea of considering a time-bounded version of Kolmogorov complexity. An interesting line of 
research~\cite{sip:c:randomness,bfl:j:boundedkolmogorov,blm:j:compression,lr:j:symmetry}, which we also pursue in this paper, focuses on the time-bounded distinguishing Kolmogorov complexity, $\cd^t(\cdot)$. We say that a program $p$ distinguishes $x$ if $p$ accepts $x$ and only $x$. $\cd^{t,A}(x)$ is the size of the smallest program that distinguishes $x$ and that runs in time $t(|x|)$ with access to the oracle $A$.  Buhrman, Fortnow, and Laplante~\cite{bfl:j:boundedkolmogorov} show that for some polynomial $p$, for every set $B$, and every string $x \in \bn$, $\cd^{p, \bn}(x) \leq 2 \log(|\bn|) + O(\log n)$. This is an important and very general result but the upper bound for the compressed string length is roughly  $2 \log(|\bn|)$ instead of $ \log(|\bn|)$, that one may hope. In fact, Buhrman, Laplante, and Miltersen~\cite{blm:c:compression},  show that for some sets $B$, the factor $2$ is necessary. There are some results where the upper bound is asymptotically $\log(|\bn|)$ at the price of weakening other parameters.
Sipser~\cite{sip:c:randomness} shows that the upper bound of $\log(|\bn|)$ can be achieved if we allow the distinguisher program to use polynomial advice: For every set $B$, there is a string $w_B$ of length $\poly(n)$ such that for every $x \in \bn$, $\cd^{poly, \bn}(x \mid w_B) \leq log (|\bn|) + \log \log(|\bn|) +O(1)$. Buhrman, Fortnow, and Laplante~\cite{bfl:j:boundedkolmogorov} show that $\log(|\bn|)$ can be achieved if we allow a few exceptions: For any $B$, any $\epsilon$, for all except a fraction of $\epsilon$ strings $x \in \bn$, $\cd^{poly, \bn}(x ) \leq log (|\bn|) + \poly\log (n \cdot 1/\epsilon)$. Buhrman, Lee, and van Melkebeek~\cite{blm:j:compression} show that for all $B$ and $x \in \bn$, $\cnd^{poly, \bn}(x ) \leq log (|\bn|) + O((\sqrt{\log(|\bn|)} + \log n) \log n)$, wher $\cnd$ is similar to $\cd$ except that the distinguisher program is nondeterministic.
\smallskip

Our main result shows that under a certain reasonable hardness assumption, the upper bound of $\log (|\bn|)$ holds for every set $B$ in PSPACE.
\smallskip

\emph{Main Result.} Assume that there exists $f \in \e$ that cannot be computed by circuits of size $2^{o(n)}$ with PSPACE gates. Then  for any $B$ in PSPACE, there exists a polynomial $p$ such that for every $x \in \bn$, 
\[
\cd^{p, \bn}(x) \leq \log(|\bn|) + O(\log n). 
\]

The main result is a corollary of the following stronger result: Under the same hardness assumption, the distinguisher program $p$ for $x$ of length $\log( |\bn|) + O(\log n)$ is simple conditioned by $x$, in the sense that $C^{\poly}(p \mid x) = O(\log n)$, where $C^{\poly}(\cdot)$ is the polynomial-time bounded Kolmogorov complexity.

We also consider some variations of the main result in which the set $B$ is in $\p$ or in $\np$. We show that the hardness assumption can be somewhat weakened by replacing the PSPACE gates with $\Sigma_3^p$ gates. We also show that the distinguisher program no longer needs oracle access to $\bn$ in case we allow it to be nondeterministic and $B$ is in $\np$.

The hardness assumption in the main result, which we call H1, states that there exists a function $f \in \e = \cup_{c\geq 0} {\rm DTIME}[2^{cn}]$ that cannot be computed by circuits of size $2^{o(n)}$ with PSPACE gates. This looks like a technical hypothesis; however,  Miltersen~\cite{mil:b:derandsurvey} shows that the more intuitive assumption ``$\e$ is not contained in ${\rm DSPACE}[2^{o(n)}]$'' implies H1. We note that this assumption (or related versions) has been used before in somewhat similar contexts. Antunes and Fortnow\cite{ant-for:c:polyunivmeasure} use a version of H1 (with the PSPACE gates replaced by $\Sigma_2^p$ gates) to show that the semi-measure $m^p(x) = 2^{-C^p(x)}$ dominates all polynomial-time samplable distributions. Trevisan and Vadhan~\cite{tre-vad:c:psamplextractor} use a version of H1 (with the PSPACE gates replaced by $\Sigma_5^p$ gates) to build for each $k$ a polynomial-time extractor for all distributions with min-entropy $(1-\delta)n$ that are samplable by circuits of size $n^k$.

\subsection{Discussion of technical aspects}
\label{s:technique}
We present the main ideas in the proof of the main result. The method is reminiscent of techniques used in the construction of Kolmogorov extractors in~\cite{fhpvw:c:extractKol,hit-pav-vin:c:Kolmextraction,zim:j:kolmextractsurvey}.  Let $B$ in PSPACE. To simplify the argument suppose that $|\bn|$ is a power of two, say $|\bn| = 2^k$. If we would have a polynomial-time computable function $E: \zo^n \mapping \zo^k$, whose restriction on $B$ is 1-to-1, then every $x \in \bn$ could be distinguished from the other elements of $\bn$ by $z = E(x)$ and we would obtain $\cd^{\poly, \bn}(x) \leq |z| + O(1) = \log(|\bn|) + O(1)$. We do not know how to obtain such a function $E$, but, fortunately, we can afford a slack additive term of $O(\log n)$ and therefore we can relax the requirements for $E$. We can consider functions $E$ of type $E: \zo^n \times \zo^{\log n} \mapping \zo^k$. More importantly, it is enough if $E$ is computable in polynomial time given an advice string $\sigma$ of length $O(\log n)$ and if every $z \in \zo^k$ has at most $O(n)$ preimages in $B \times \zo^{\log n}$. With such an $E$, the string $z = E(x, 0^{\log n})$ distinguishes $x$ from strings that do not map into $z$ and, using the general result of Buhrman, Fortnow, and Laplante~\cite{bfl:j:boundedkolmogorov}, with additional $2 \log n + O(1)$ bits we can distinguish $x$ from the other at most $O(n)$ strings that map into $z$. With such an $E$, we obtain for every $x \in \bn$ the desired $\cd^{\poly, \bn}(x) \leq |z| + |\sigma| + 2 \log n +O(1) = \log(|\bn|) + O(\log n)$.

Now it remains to build the function $E$. An elementary use of the probabilistic method shows that if we take $E: \zo^n \times \zo^{\log n} \mapping \zo^k$ at random, with high probability, every $z \in \zo^k$ has at most $7n$ preimages. The problem is that to compute a random $E$ in polynomial-time we would need its table and the table of such a function has size $\poly(N)$, where $N=2^n$. This is double exponentially larger than $O(\log n)$ which has to be the size of $\sigma$ from our discussion above.

To reduce the size of advice information (that makes  $E$ computable in polynomial time) from $\poly(N)$ to $O(\log n)$, we derandomize the probabilistic construction in two steps.

In the first step we observe that counting (the number of preimages of $z$) can be done with sufficient accuracy by circuits of size $\poly(N)$ and constant-depth using the result of Ajtai~\cite{ajt:j:constantdepthcount}. We infer that there exists a circuit $G$ of size $\poly(N)$ and constant depth such that $\{E \mid \mbox{ every $z$ has $\leq 7n$ preimages in $B \times \zo^{\log n}$}\} \subseteq \{ E \mid G(E) =1 \} \subseteq \{E \mid \mbox{ every $z$ has $\leq 8n$ preimages in $B \times \zo^{\log n}$}\}$.
Now we can utilize the Nisan- Wigderson~\cite{nis-wig:j:hard} pseudo-random generator $\mbox{NW-gen}$ against constant-depth circuits and we obtain that, for most seeds $s$ (which we call good seeds for $\mbox{NW-gen}$), $\mbox{NW-gen}(s)$ is the table of a function $E$ where each element $z \in \zo^k$
has at most $8n$ preimages in $B \times \zo^{\log n}$. This method is inspired by the work of Musatov~\cite{mus:t:spacekolm}, and it has also been used in~\cite{zim:c:symkolm}.  The seed $s$ has size $\poly \log (N) = \poly(n)$, which is not short enough.

In the second step we use the Impagliazzo-Wigderson pseudo-random generator~\cite{imp-wig:c:pbpp} as generalized by Klivans and van Melkebeek~\cite{km:j:prgenoracle}. We observe that checking that a seed $s$ is good for $\mbox{NW-gen}$ can be done in PSPACE, and we use the hardness assumption to infer the existence of a pseudo-random generator $H$ such that for most seeds $\sigma$ of $H$ (which we call good seeds for $H$), $H(\sigma)$ is a good seed for $\mbox{NW-gen}$.
We have $|\sigma| = \log |s| = O(\log n)$ as desired. Finally, we take our function $E$ to be the function whose table is $\mbox{NW-gen}(H(\sigma))$, for some good seed $\sigma$ for $H$. It follows that, given $\sigma$, $E$ is computable in polynomial time and that every $z \in \zo^k$ has at most $8n$ preimages in $\bn \times \zo^{\log n}$. 

The idea of the $2$-step derandomization has been used by Antunes and Fortnow~\cite{ant-for:c:polyunivmeasure} and later by Antunes, Fortnow, Pinto and Souza~\cite{afps:c:lowdepthwit}.

\section{Preliminaries}
\subsection{Notation and basic facts on Kolmogorov complexity}
We work over the binary alphabet $\zo$. A string is an element of $\{0,1\}^*$.
If $x$ is a string, $|x|$ denotes its length; if $B$ is a finite set, $|B|$ denotes its size.  If $B \subseteq \zo^*$, then $\bn = \{x \in B \mid |x| = n\}$.

The Kolmogorov complexity of a string $x$ is the length of the shortest program that prints $x$. 
The $t$-time bounded Kolmogorov complexity of a string $x$ is the length of the shortest program that prints $x$ in at most $t(|x|)$ steps. The $t$-time bounded distinguishing Kolmogorov complexity of a string $x$ is the length of the shortest program that accepts $x$ and only $x$ and runs in at most $t(|x|)$ steps. The formal definitions are as follows.

We fix an universal Turing machine $U$, which is able to simulate any other Turing machine with only a constant additive term overhead in the program length. The Kolmogorov complexity of the string $x$ conditioned by string $y$, denoted $C(x \mid y)$, is the length of the shortest  string $p$ (called a program) such that $U(p,y) = x$. In case $y$ is the empty string, we write $C(x)$. 

For the time-bounded versions of Kolmogorov complexity, we fix an universal machine $U$, that, in addition to the above property, can also simulate any Turing machine $M$ in time $t_M(|x|) \log t_M(|x|)$, where $t_M(|x|)$ is the running time of $M$ on input $x$. For a time bound $t(\cdot)$, the $t$-bounded Kolmogorov complexity of $x$ conditioned by $y$, denoted $C^t(x \mid y)$, is the length of the shortest program $p$ such that $U(p,y) = x$ and $U(p,y)$ halts in at most $t(|x|+ |y|)$ steps.

The $t$-time bounded distinguishing complexity of $x$ conditioned by $y$, denoted $\cd^t(x \mid y)$ is the length of the shortest program $p$ such that

(1) $U(p,x,y)$ accepts,

(2) $U(p,v,y)$ rejects for all $v \not= x$,

(3) $U(p,v,y)$ halts in at most $t(|v| + |y|)$ steps for all $v$ and $y$.

In case $y$ is the empty string $\lambda$, we write $\cd^t(x)$ in place of $\cd^t(x \mid \lambda)$.  If $U$ is an oracle machine, we define in the similar way, $\cd^{t,A}(x \mid y)$ and $\cd^{t,A}(x)$, by allowing $U$ to query the oracle $A$.

For defining $t$-time bounded nondeterministic distinguishing Kolmogorov complexity, we fix  $U$ a nondeterministic universal machine, and we define $\cnd^{t} (x \mid y)$ in the similar way.

\if01
There are several versions of this notion. We use here  the \emph{plain complexity}, denoted $C(x)$, and also the \emph{conditional plain complexity} of a string $x$ given a string $y$, denoted $C(x \mid y)$, which is the length of the shortest effective description of $x$ given $y$. The formal definitions are as follows.
We work over the binary alphabet $\zo$. A string is an element of $\{0,1\}^*$.
If $x$ is a string, $|x|$ denotes its length.  
Let $M$ be a Turing machine that takes two input strings and outputs one string. For any strings $x$ and $y$, define the \emph{Kolmogorov complexity} of $x$ conditioned by $y$ with respect to $M$, as 
$C_M(x \mid y) = \min \{ |p| \mid M(p,y) = x \}$.
There is a universal Turing machine $U$ with the following property: For every machine $M$ there is a constant $c_M$ such that for all $x$, $C_U(x \mid y) \leq C_M(x \mid y) + c_M$.
We fix such a universal machine $U$ and dropping the subscript, we write $C(x \mid y)$ instead of $C_U(x \mid y)$. We also write $C(x)$ instead of $C(x \mid \lambda)$ (where $\lambda$ is the empty string). The \emph{randomness rate} of a string $x$ is defined as ${\rm rate}(x) = \frac{C(x)}{|x|}$.  If $n$ is a natural number, $C(n)$ denotes the Kolmogorov complexity of the binary representation of $n$. For two $n$-bit strings $x$ and $y$, the information in $x$ about $y$ is denoted $I(x : y)$ and is defined as $I(x : y) = C(y \mid n) - C(y \mid x)$.

In this paper, the constant hidden in the $O(\cdot)$ notation only depends on the universal Turing machine.

For all $n$ and $k \leq n$, $2^{k-O(1)} < |\{x \in \zo^n \mid C(x\mid n) < k\}| < 2^k$.

%\item By simple counting arguments it can be shown that for every %sufficiently large $n$ and $k \leq n$,
%$2^{k-2\log n} < |\{x \in \zo^n \mid C(x) \leq k\}| < 2^{k+1}$. 
\fi

Strings $x_1, x_2, \ldots, x_k$ can be encoded in a self-delimiting way (\ie, an encoding from which each string can be retrieved) using $|x_1| + |x_2| + \ldots + |x_k| + 2 \log |x_2| + \ldots + 2 \log |x_k| + O(k)$ bits. For example, $x_1$ and $x_2$ can be encoded as $\overline{(bin (|x_2|)} 01 x_1 x_2$, where $bin(n)$ is the binary encoding of the natural number $n$ and, for a string $u = u_1 \ldots u_m$, $\overline{u}$ is the string $u_1 u_1 \ldots u_m u_m$ (\ie, the string $u$ with its bits doubled).

\subsection{Distinguishing complexity for strings in an arbitrary set}
As mentioned,  Buhrman, Fortnow and Laplante~\cite{bfl:j:boundedkolmogorov}, have shown that for any set $B$ and for every $x \in \bn$ it holds that
$\cd^{\poly, \bn} (x) \leq 2 \log (|\bn|) + O(\log n)$. We need  the following more explicit statement of their result.
\begin{theorem}
\label{t:bfl}
There exists a polynomial-time algorithm $A$ such that for every set $B \subseteq \zo^*$, every $n$, every $x \in \bn$, there exists a string $p_x$ of length $|p_x| \leq 2 \log(|\bn|) + O(\log n)$ such that 
\begin{itemize}
	\item [$\bullet$ ] $A(p_x, x) = $ accept,
	\item [$\bullet$ ] $A(p_x, y) = $ reject, for every $y \in \bn - \{x\}$.
\end{itemize}

\end{theorem}

\subsection{Approximate counting via polynomial-size constant-depth circuits}
We will need to do counting with constant-depth polynomial-size circuits. Ajtai~\cite{ajt:j:constantdepthcount} has shown that this can be done with sufficient precision.

\begin{theorem}
\label{t:ajtai}
({\bf{Ajtai's approximate counting with polynomial size constant-depth circuits}})
There exists a uniform family of circuits $\{G_n\}_{n \in \nat}$, of polynomial size and constant depth, such that for every $n$, for every $x \in \zo^n$, for every $a \in \{0, \ldots, n-1\}$, and for every $\epsilon > 0$,
\begin{itemize}
	\item [$\bullet$ ] If the number of $1$'s in $x$ is $\leq (1 - \epsilon)a$, then $G_n(x,a,1/\epsilon) = 1$,
	\item [$\bullet$ ] If the number of $1$'s in $x$ is $\geq (1 + \epsilon)a$, then $G_n(x,a,1/\epsilon) = 0$.
\end{itemize}
\end{theorem}
 We do not need the full strength (namely, the uniformity of $G_n$) of this theorem; the required level of accuracy (just $\epsilon > 0$) can be achieved by non-uniform polynomial-size circuits of depth $d=3$ (with a much easier proof, see~\cite{vio:t:approxcount}).

\subsection{Pseudo-random generator fooling bounded-size constant-depth circuits}
The first step in the derandomization argument in the proof of Theorem~\ref{t:main} is done using the Nisan-Wigderson pseudo-random generator that ``fools" constant-depth circuits~\cite{nis-wig:j:hard}. 
\begin{theorem}[Nisan-Wigderson pseudo random generator]
\label{t:NWgen}
For every constant $d$ there exists a constant $\alpha > 0$ with the following property. There exists a function $\mbox{NW-gen}:\zo^{O(\log^{2d+6}n)} \mapping \zo^n$ such that for any circuit $G$ of size $2^{n^\alpha}$ and depth $d$,
\[
| \prob_{s \in \zo^{O(\log^{2d+6}n)}}[G(\mbox{NW-gen}(s)) = 1] - \prob_{z \in \zo^n} [G(z)=1]| < 1/100.
\]
Moreover, there is a procedure that on inputs $(n, i,s)$ produces the $i$-th bit of $\mbox{NW-gen}(s)$ in time
$\poly(\log n)$.

\end{theorem}
\subsection{Hardness assumptions and pseudo-random generators}
The second step in the derandomization argument in the proof of Theorem~\ref{t:main} uses  pseudo-random generators based on the assumption that hard functions exist in $\e = \cup_c {\rm DTIME}[2^{cn}]$. Such pseudo-random generators were
constructed by Nisan and Wigderson~\cite{nis-wig:j:hard}. Impagliazzo and Wigderson~\cite{imp-wig:c:pbpp} strenghten the results in~\cite{nis-wig:j:hard} by weakening the hardness assumptions. Klivans and van Melkebeek~\cite{km:j:prgenoracle} show that the Impagliazzo-Wigderson results hold in relativized worlds. We use in this paper some instantiations of a general result in~\cite{km:j:prgenoracle}.

We need the following definition. For a function $f: \zo^* \mapping \zo$ and an oracle $A$, the circuit complexity $C_f^A(n)$ of $f$ at length $n$ relative to $A$ is the smallest integer $t$ such that there exists an $A$ oracle circuit of size $t$ that computes $f$ on inputs of length $n$. 

We use the following hardness assumptions.

{\bf Assumption H1:}
\smallskip

 There exists $f \in \e$ such that for some $\epsilon > 0$ and for some PSPACE complete set $A$, $C_f^{A}(n) \geq 2^{\epsilon \cdot n}$.
 \medskip
 
{\bf Assumption H2:}
\smallskip

 There exists $f \in \e$ such that for some $\epsilon > 0$ and for some $\Sigma_3^p$ complete set $A$, $C_f^A(n) \geq 2^{\epsilon \cdot n}$.
 \smallskip
 
 If H1 holds, then for some oracle $A$ that is PSPACE complete, for every $k$, there exists $H: \zo^{c \log n} \mapping \zo^n$ computable in time $\poly(n)$ such that for every oracle circuit $C$ of size $n^k$, 
 \[
 | \prob_{\sigma \in \zo^{c \log n}}[C^A(H(\sigma))=1] - \prob_{s \in \zo^n}[C^A(s) = 1]| < o(1).
 \]
 If H2 holds, then for some oracle $A$ that is $\Sigma_3^p$ complete, for every $k$, there exists $H: \zo^{c \log n} \mapping \zo^n$ computable in time $\poly(n)$ such that for every oracle circuit $C$ of size $n^k$, 
 \[
 | \prob_{\sigma \in \zo^{c \log n}}[C^A(H(\sigma))=1] - \prob_{s \in \zo^n}[C^A(s) = 1]| < o(1).
 \]
 
\section{Main Result}

\begin{theorem}
\label{t:maintechnical}
Assuming H1, the following holds: For every set $B$ in PSPACE, there exists a polynomial $p$ such that for every length $n$, and for every  string $x \in \bn$, there exists a string $z$ with the following properties:

(1) $|z| = \lceil \log (|\bn|) \rceil$,

(2) $C^p(z \mid x) = O(\log n)$,

(3) $\cd^{p, \bn} (x \mid z) = O(\log n)$.
\end{theorem}
Before proving the theorem, we note that (1) and (3) immediately imply the following theorem, which is our main result.
\begin{theorem}
\label{t:main}
Assumming H1, the following holds: For every set $B$ in PSPACE there exists a polynomial $p$, such that  for every length $n$, and for every  string $x \in \bn$,
\[
\cd^{p, \bn} (x) \leq \log(|\bn|) + O(\log n).
\]
\end{theorem}
\begin{proof} (of Theorem~\ref{t:maintechnical}) 
Let us fix a set $B$ in PSPACE and let us focus on $\bn$, the set of strings of length $n$ in $B$. Let $k = \lceil \log |\bn| \rceil$ and let $K = 2^k$. Also, let $N = 2^n$.
 \begin{definition}
 \label{d:balanced}
 Let $E : \zo^n \times \zo^{\log n} \mapping \zo^k$. We say that $E$ is $\Delta$-balanced  if for every $z \in \zo^m$,
 \[
 \big | E^{-1}(z) ~\cap~ \bn \times \zo^{\log n} \big| \leq \Delta \cdot \frac{|\bn| \cdot n}{K}.
 \]
 \end{definition}
 
 The plan for the proof is as follows.  Suppose that we have a function $E: \zo^n \times \zo^{\log n} \mapping \zo^k$ that is $\Delta$-balanced, for some constant $\Delta$.
 
 Furthermore assume that $E$ can be ``described'' by a string $\sigma$, in the sense that given $\sigma$ as an advice string,  $E$ is computable in time polynomial in $n$. 
 
 Fix $x$ in $\bn$. and let $z = E(x, 0^{\log n})$. Clearly, the string $z$ satisfies requirement (1). It remains to show (2) and (3).
 
 Consider the set
 \[
 U=\{u \in \bn \mid \exists v \in \zo^{\log n}, E(u,v) = z\}.
 \]
 
 Since $E$ is $\Delta$-balanced, the size of $U$ is bounded by $\Delta \cdot \frac{|\bn|\cdot n}{K} \leq \Delta n$.
 
 Now observe that for some polynomial $p$,
 \[
 \cd^{p, \bn}(x \mid z) \leq |\sigma| + 2 \log(\Delta n) + O(\log n) +  \mbox{self-delimitation overhead}.
 \]
 Indeed, the following is a polynomial-time algorithm using oracle $\bn$ that distinguishes $x$ (it uses an algorithm $A$, promised by Theorem~\ref{t:bfl},
 that distinguishes $x$ from the other strings in $U$, using a string $p_x$ of length $2 \log(|U^{=n}|) + O(\log n) \leq 2 \log (\Delta n) + O(\log n)$).
 \smallskip

 \fbox{\vbox{
 Input: $y$; the strings $z, \sigma, p_x$, defined above, are also given.
 
 If $y \not \in \bn$, then reject.
 
 If for all $v \in \zo^{\log n}$, $E(y, v) \not= z$, then reject.
 
 If $A(y, p_x) = $ reject, then reject.
 
 Else accept.
 }}
 \smallskip

Clearly, the algorithm accepts input $y$ iff $y = x$.  Also, since $z = E(x, 0^{\log n})$, $C^p(z \mid x) \leq |\sigma| + O(1)$. For a further application (Theorem~\ref{t:cdnptechnical}), note that the above algorithm queries the oracle $\bn$ a single time.
 
 Therefore, if we manage to achieve $\sigma = O(\log n)$, we obtain that $\cd^{p, \bn}(x \mid z ) \leq O(\log n)$ and 
 $C^p(z \mid x) \leq O(\log n)$.
 
 Thus our goal is to produce a function $E: \zo^n \times \zo^{\log n} \mapping \zo^k$ that using an advice string $\sigma$ of length $O(\log n)$ is computable in polynomial time and is $\Delta$-balanced for some constant $\Delta$.
 Let us call this goal $(*)$.
 
 We implement the ideas exposed in Section~\ref{s:technique} and the reader may find convenient to review that discussion.
 
 \begin{myclaim}.
 \label{c:probmethod}
With probability at least $0.99$, a random function $E: \zo^n \times \zo^{\log n} \mapping \zo^k$ is $7$-balanced.
 
 \end{myclaim}
 \begin{proof}
 For fixed $x \in B$, $y \in \zo^{\log n}$, $z \in \zo^m$, if we take a random function $E : \zo^n \times \zo^{\log n} \mapping \zo^k$, we have that $\prob[E(x,y) = z] = 1/K$.
 Thus the expected number of preimages of $z$ in the rectangle $B \times \zo^{\log n}$ is $\mu = (1/K) \cdot |B| \cdot n$.
 By the Chernoff's bounds,\footnote{We use the following version of the Chernoff bound. If $X$ is a sum of independent Bernoulli random variables, and the expected value $E[X]= \mu$, then
$\prob[X \geq (1+\Delta)\mu] \leq e^{-\Delta (\ln (\Delta/3)) \mu}$. The standard Chernoff inequality $\prob(X \geq (1+\Delta) \mu] \leq \big( \frac{e^\Delta}{(1+\Delta)^{(1+\Delta)}}\big)^\mu$ is presented in many textbooks. It can be checked easily that $\frac{e^\Delta}{(1+\Delta)^{(1+\Delta)}} < e^{-\Delta \ln (\Delta/3)}$, which implies the above form.}
 \[
 \prob[\mbox{number of preimages of $z$ in $B \times \zo^{\log n}$} > 7  \mu] < e^{-(6 \ln 2) \mu}.
 \]
 Therefore, the probability of the event ``there is some $z \in \zo^k$ such that the number of $z$-cells in $B \times \zo^{\log n}$ is $ > 7 \mu$'' is at most $K \cdot e^{-(6 \ln 2) \mu} < 0.01$.
\qed \end{proof}
\begin{myclaim}.
\label{c:circuit}
There exists a circuit $G$ of size $\poly(N)$ and constant depth such that for any function $E: \zo^n \times \zo^{\log n} \mapping \zo^k$ (whose table is given to $E$ as the input),
\begin{itemize}
	\item [(a)] If $G(E)=1$, then $E$ is $8$-balanced,
	\item [(b)] If $E$ is $7$-balanced, then $G(E)=1$.
\end{itemize}

\end{myclaim}
 \begin{proof}
By Theorem~\ref{t:ajtai}, there is a constant-depth, $\poly(N)$ size circuit that counts in an approximate sense the occurrences of a string $z$ in $\zo^k$ in the rectangle $B \times \zo^{\log n}$. If we make a copy of this circuit for each $z \in \zo^k$ and link all these copies to an AND gate we obtain the desired circuit $G$.

More precisely, let $x_z$ be the binary string of length $|B| \cdot n$, whose bits are indexed as $(u,v)$ for $u \in B, v \in \zo^{\log n}$, and whose $(u,v)$-bit is $1$ iff $E(u,v)=z$. By Theorem~\ref{t:ajtai}, there is a constant-depth, $\poly(N)$ size circuit $G'$ that behaves as follows:
\begin{itemize}
	\item [$\bullet$ ] $G'(x_z) = 1$ if the number of $1$'s in $x_z$ is $\leq 7 \frac{|B|\cdot n}{K}$,
	\item [$\bullet$ ] $G'(x_z) = 0$ if the number of $1$'s in $x_z$ is $> 8 \frac{|B|\cdot n}{K}$,
	\end{itemize}
 If the number of $1$'s is between the two bounds, the circuit $G'$ outputs either $0$ or $1$.
 
 The circuit $G$ on input a table of $E$, will first build the string $x_z$ for each $z \in \zo^k$, then has a copy of $G'$ for each $z$, with the $z$-th copy running on $x_z$ and then connects the outputs of all the copies to an AND gate, which is the output gate of $G$. 
 
 \qed \end{proof}
 \begin{myclaim}.
 \label{c:ge}
 If we pick at random a function $E: \zo^n \times \zo^{\log n} \mapping \zo^k$, with probability at least $0.99$, $G(E) = 1$.
 \end{myclaim}
 \begin{proof}
 This follows from Claim~\ref{c:probmethod} and from Claim~\ref{c:circuit} (b).
 \qed \end{proof}
 
 Let $\tilde{N} = N \cdot n \cdot k$. Let $d$ be the depth of the circuit $G$. We denote $\tilde{n} = log^{2d+6}\tilde{N}$. Note that $\tilde{n} = \poly(n)$. We consider the Nisan-Wigderson pseudo-random generator for depth $d$ given by Theorem~\ref{t:NWgen}. Thus,
 \[
 \mbox{NW-gen}: \zo^{\tilde{n}} \mapping \zo^{\tilde{N}}.
 \]
 For any string $s$ of length $\tilde{n}$, we view
 $\mbox{NW-gen}(s)$ as the table of a function $E : \zo^n \times \zo^{\log n} \mapping \zo^k$.
 \begin{myclaim}.
 If we pick at random $s \in \tilde{n}$, with probability of $s$ at least $0.9$, it holds that 
 $\mbox{NW-gen}(s) \mbox{ is $8$-balanced}$.
 \end{myclaim}
 \begin{proof}
 Since $G$ is a circuit of constant depth and polynomial size, by Theorem~\ref{t:NWgen}, the probability of the event
 ``$G(\mbox{NW-gen}(s)) = 1$'' is within $0.01$ of the probability of the event ``$G(E)=1$,'' and the second probability is at least $0.99$ by Claim~\ref{c:ge}. Thus the first probability is greater than $0.9$. Taking into account that if 
 $G(\mbox{NW-gen}(s)) = 1$ then $\mbox{NW-gen}(s)$ is $8$-balanced, the conclusion follows.
 \qed \end{proof}
 
 \begin{myclaim}.
 Let $T = \{s \in \zo^{\tilde{n}} \mid \mbox{NW-gen}(s) \mbox{ is $8$-balanced} \}$. Then $T$ is in PSPACE.
  \end{myclaim}
 \begin{proof}
 We need to count for every $z \in \zo^k$, the number of $z$-cells in the rectangle $\bn \times \zo^{\log n}$ of the table of $\mbox{NW-gen}(s)$. Since $B$ is in PSPACE and since each entry in the table of $\mbox{NW-gen}(s)$ can be computed in time polynomial in $\tilde{n}$, it follows that the above operation can be done in PSPACE.
 \qed \end{proof}
 \begin{myclaim}.
 \label{f:secondprg}
 Assume H1. There exists a constant $c$ and a function $H : \zo^{c \log \tilde{n}} \mapping \zo^{\tilde{n}}$, computable in time $\poly(\tilde{n}) = \poly(n)$, such that if $\sigma$ is a string chosen at random in $\zo^{c \log \tilde{n}}$, with probability at least $0.8$, it holds that $\mbox{NW-gen}(H(\sigma)))$ is $8$-balanced.
 \end{myclaim}
 \begin{proof}
 Under assumption H1, there exists a pseudo-random generator $H: \zo^{c \log \tilde{n}} \mapping \zo^{\tilde{n}}$ such that for any set $A$ in PSPACE,
 \[
 |\prob_{\sigma \in \zo^{c \log \tilde{n}}}[H(\sigma) \in A] - \prob_{s \in \zo^{\tilde{n}}}[ s \in A]| < 0.1.
 \]
 Since the set $T$ is in PSPACE, 
 $\prob_{\sigma \in \zo^{c \log \tilde{n}}}[\mbox{NW-gen}(H(\sigma)) \mbox{ is  $8$-balanced}]$
is within $0.1$ from 
$\prob_{s \in \zo^{\tilde{n}}}[ \mbox{NW-gen}(s) \mbox{ is  $8$-balanced}].$
 Since the latter probability is at least $0.9$, the conclusion follows.
 
 \qed \end{proof}
 \smallskip
 
 We can now finish the proof of Theorem~\ref{t:main}.
 
 Fix $\sigma \in \zo^{c \log \tilde{n}}$ such that $\mbox{NW-gen}(H(\sigma)))$ is $8$-balanced.  Note that $|\sigma| = O(\log n)$. Given $\sigma$, each entry in the table of $\mbox{NW-gen}((H(\sigma)))$ can be computed in time $\poly(n)$. Thus the function $E: \zo^n \times \zo^{\log n} \mapping \zo^k$, whose table is $\mbox{NW-gen}((H(\sigma)))$, satisfies the goal $(*)$.
 \qed \end{proof}
 
 \section{Variations around the main result}
 
 We analyze here the polynomial-time bounded distinguishing Kolmogorov of strings in a set $B$ that is in $\p$ or in $\np$. We start with the case $B \in \p$. The following is the analog of Theorem~\ref{t:maintechnical} and its main point is that assumption H1 can be replaced by the presumably weaker assumption H2.
 \begin{theorem}
\label{t:cdptechnical}
Assuming H2, the following holds:  For every set $B$ in $\p$, there exists a polynomial $p$ such that, for every length $n$, and for every  string $x \in \bn$, there exists a string $z$ with the following properties:

(1) $|z| = \lceil \log (|\bn|) \rceil$,

(2) $C^p(z \mid x) = O(\log n)$,

(3) $\cd^{p} (x \mid z) = O(\log n)$.
\end{theorem}

 \begin{proof}
 We follow the proof of Theorem~\ref{t:maintechnical}. First note that since $B \in \p$, the universal machine does not need oracle access to $B$. We still need to justify that assumption H1 can be replaced by the weaker assumption H2.
 
 Assumption H1 was used in Claim~\ref{f:secondprg}. The point was that the set $T=\{s \mid \mbox{NW-gen}(s) \mbox{ is $8$-balanced}\}$ is in PSPACE and H1 was used to infer the existence of a pseudo-random generator $H$ that fools $T$. If $B \in \p$, we can check that $\mbox{NW-gen}(s)$ is sufficiently balanced using less computational power than PSPACE.
 Basically we need to check that for all $z \in \zo^k$,
 \[
 |\mbox{NW-gen}(s)^{-1}(z) \cap \bn \times \zo^{\log n}| \leq \Delta \cdot \frac{|\bn| \cdot n}{K},
 \]
 for some constant $\Delta$. Using Sipser's method from~\cite{sip:c:randomness}, there is a $\Sigma_2^p$ predicate $R$ such that
 
\begin{itemize}
	\item [$\bullet$] $R(s,z)=1$ implies  
$|\mbox{NW-gen}(s)^{-1}(z) \cap \bn \times \zo^{\log n}| \leq 16 \cdot n$,
	\item [$\bullet$]
	$R(s,z)=0$ implies  
$|\mbox{NW-gen}(s)^{-1}(z) \cap \bn \times \zo^{\log n}| \geq 64 \cdot n$.
\end{itemize}
Thus there is a set $T' \subseteq \zo^{\tilde{n}}$ in $\Sigma_3^p$ such that for all $s \in T'$, $\mbox{NW-gen}(s)$ is $64$-balanced and $T'$ contains all $s$ such that $\mbox{NW-gen}(s)$ is $8$-balanced. Note that the second property implies that $|T'| \geq 0.99 \cdot 2^{\tilde{n}}$.

Now assumption H2 implies that there exists a pseudo-random generator $H: \zo^{c \log(\tilde{n})} \mapping \zo^{\tilde{n}}$ that fools $T'$. In particular it follows that with probability of $\sigma \in \zo^{c \log (\tilde{n})}$ at least $0.8$, $H(\sigma) \in T'$ and thus $\mbox{NW-gen}(H(\sigma))$ is $64$-balanced. The rest of the proof is identical with the proof of Theorem~\ref{t:main}.
 
 \qed \end{proof}
 The next result is the analog of Theorem~\ref{t:main} for the case when the set $B$ is in $\p$.
  \begin{theorem}
 \label{t:cdp}
 Assuming H2, the following holds: For every $B \in \p$, there exists a polynomial $p$ such that for all $n$, and for all $x \in \bn$, 
 \[
 \cd^{poly}(x) \leq \log(|\bn|) + O(\log n).
 \]
 \end{theorem}
 \begin{proof}
 This is an immediate consequence of (1) and (3) in Theorem~\ref{t:cdptechnical}.
\qed \end{proof}
Next we consider the case when the set $B$ is in $\np$. The main point is that the assumption H1 can be replaced by H2, and that the distinguishing program does not need access to the oracle $\bn$ provided it is nondeterministic.
 \begin{theorem}
\label{t:cdnptechnical}
Assuming H2, the following holds: For every set $B$ in $\np$, there exists a polynomial $p$ such that  for every length $n$, and for every  string $x \in \bn$, there exists a string $z$ with the following properties:

(1) $|z| = \lceil \log (|\bn|) \rceil$,

(2) $C^p(z \mid x) = O(\log n)$,

(3) $\cd^{p, \bn} (x \mid z) = O(\log n)$.

(4) $\cnd^{p} (x \mid z) = O(\log n)$.
\end{theorem}
 \begin{proof}
 (1), (2) and (3). We only need to show that in the proof of Theorem~\ref{t:maintechnical}, in case $B \in \np$, the assumption H1 can be replaced by the weaker assumption H2. This is done virtually in the same way as in the proof of Theorem~\ref{t:cdp}. The predicate $R$ also needs this time to check that certain strings are in $B$ and this involves an additional quantifier, but that
 quantifier can be merged with the existing quantifiers and $R$ remains in $\Sigma_2^p$.
 \smallskip
 
 (4). We need to show that, at the price of replacing $\cd$ by $\cnd$, the use of the oracle $\bn$ is no longer necessary. Note that the distinguisher procedure given in the proof of Theorem~\ref{t:maintechnical}, queries the oracle only once, and if the answer to that query is NO, then the algorithm rejects immediately. Thus, instead of making the query, a
 nondeterministic distinguisher can just guess a witness for the single query it makes.
 \qed \end{proof}
 The following is the analog of Theorem~\ref{t:main} in case the set $B$ is in $\np$.
 \begin{theorem}
 \label{t:cdnp}
  Assuming H2, the following holds: 
  
  (a) For every $B \in \np$, there exists a polynomial $p$, such that for all $n$, and for all $x \in \bn$, 
 \[
 \cd^{p, \bn}(x) \leq \log(|\bn|) + O(\log n).
 \]
 
 (b) For every $B \in \np$, there exists a polynomial $p$, such that for all $n$, and for all $x \in \bn$, 
 \[
 \cnd^{p}(x) \leq \log(|\bn|) + O(\log n).
 \]
 \end{theorem}
 \begin{proof}
 Statement (a) follows from (1) and (3) in Theorem~\ref{t:cdnptechnical}, and (b) follows from (1) and (4) in Theorem~\ref{t:cdnptechnical}.
 \qed \end{proof}
 \bibliography{c:/book-text/theory}

\bibliographystyle{alpha}

\newpage

\appendix
 
\end{document}